\documentclass[a4paper,UKenglish,cleveref, autoref, thm-restate]{lipics-v2021}

\pdfoutput=1 

\title{Computing Short SAT Implicants via Ising/QUBO Encodings}

\author{Giuseppe Spallitta}{Rice University, Houston, Texas, United States of America}{gs81@rice.edu}{https://orcid.org/0000-0002-4321-4995}{}

\author{Leonardo Duenas-Osorio}{Rice University, Houston, Texas, United States of America}{lad1@rice.edu}{https://orcid.org/0000-0002-7138-7746}{}

\author{Moshe Y. Vardi}{Rice University, Houston, Texas, United States of America}{vardi@rice.edu}{https://orcid.org/0000-0002-0661-5773}{}

\authorrunning{G. Spallitta, L. Duenas-Osorio and M. Y. Vardi} 

\Copyright{G. Spallitta, L. Duenas-Osorio and M. Y. Vardi} 



\category{} 

\relatedversion{} 

\supplementdetails[
subcategory={}
]{Software}{https://github.com/giuspek/minimal_assignment_via_Ising} 

\nolinenumbers 

\usepackage{xspace}
\newcommand{\defas}{\ensuremath{\stackrel{\text{\scalebox{.7}{def}}}{=}}\xspace}

\EventEditors{Nicolas Beldiceanu}
\EventNoEds{1}
\EventLongTitle{32nd International Conference on Principles and Practice of Constraint Programming (CP 2026)}
\EventShortTitle{CP 2026}
\EventAcronym{CP}
\EventYear{2026}
\EventDate{July 20--23, 2026}
\EventLocation{Lisbon, Portugal}
\EventLogo{}
\SeriesVolume{379}
\ArticleNo{25}

\begin{document}

\maketitle

\begin{abstract}
Many reasoning tasks require short partial satisfying assignments (implicants), sometimes focusing on a set of important variables. SAT-to-Ising-QUBO formulations are implicitly designed so that ground states correspond to total assignments, since the Ising/QUBO model assigns a value to every spin and has no native representation of unassigned variables.
We introduce an Ising/QUBO framework that incorporates "don’t-care" semantics into the quadratic model via a dual-polarity representation, enabling the retrieval of short implicants. The encoding supports implicant shrinking and projection through minor objective modifications. We provide parameter regimes under which ground states correspond to short partial satisfying assignments, achieving minimality and, when the quadratic penalty function permits, minimum-cardinality. We empirically evaluate the encoding with simulated annealing on random 3-SAT enumeration benchmarks and non-CNF formulas, showing that it leaves about one-third of variables unassigned on random 3-SAT formulas while preserving satisfiability, and that consecutive polarity-freezing rounds achieve minimality (and minimum-cardinality) with high probability.
\end{abstract}

\section{Introduction}

Many reasoning and knowledge-representation tasks require \emph{short partial satisfying assignments} (implicants), where unassigned variables are genuine don’t-cares \cite{ignatiev2016finding}. In many applications, only a subset of variables is semantically relevant, and one seeks short partial assignments \emph{over that subset} while allowing the remaining variables to act as hidden support. Such (projected) partial models are central to explanation and diagnosis~\cite{MarquesSilvaJanota2013}, model compression~\cite{Moehle2025}, and enumeration pipelines that compactly represent large solution spaces via projection~\cite{masina2023cnf,masina2025cnf,masina2025exploiting}. More broadly, the SAT literature has long emphasized reasoning with partial assignments in contexts such as abstraction and logic minimization~\cite{Moehle2025, MarquesSilvaJanota2013}.
Dealing with partiality, however, is not trivial. For simplicity and efficiency, modern SAT solvers naturally compute total satisfying assignments, and obtaining short implicants typically requires repeated SAT calls or dedicated post-processing~\cite{deharbe2013computing}. 

Interestingly, the same phenomenon appears in SAT-to-Ising/QUBO encodings, i.e., translations that map a CNF formula into a quadratic energy function whose ground states correspond to satisfying assignments. SAT and MaxSAT have been extensively studied in the QUBO and Ising frameworks, motivated by the evolution of quantum annealing and classical Ising heuristics~\cite{Lucas2014,bian2017solving,chancellor2016direct,nusslein2023solving,Zielinski2023}. Still, the logical interpretation of the optimizer output typically corresponds to a total Boolean assignment. This is inherent to the Ising setting, where optimization assigns a value to each Ising variable, thereby providing no native representation of ``unassigned'' values. As a result, computing short implicants is not captured within the quadratic model itself and is instead left to external algorithms.

This paper introduces a \emph{unified optimization framework} for computing short partial assignments of CNF formulas in Ising/QUBO form, with native support for implicant shrinking and projection. The key idea is to represent, within the quadratic model, the distinction between variables set to false and unassigned variables, so that a QUBO optimizer's output can be decoded into a partial Boolean assignment. Within this framework, implicant shrinking and projection do not require constructing a distinct optimization model: they are obtained by minor modifications to the original formulation, resulting in a flexible approach to computing partial assignments via quadratic optimization. The encoding could then be used as a drop-in replacement for shrinking subroutines in AllSAT/explanation pipelines \cite{spallitta2025disjoint}.

An important theoretical question is under which weight regimes the framework correctly enforces clause satisfaction, partial-assignment semantics, and minimum-cardinality. We therefore identify conditions ensuring that the implicants decoded from ground states have minimum cardinality, and provide a simple criterion for recognizing implicants from the raw QUBO output energy.

\textit{Remark.} Our experiments are intended as a semantic validation of the encoding: we test whether low-energy configurations decode to short (possibly minimum) partial assignments and whether the expected behavior is visible already under a generic optimization backend.

\paragraph*{Contributions.} 

\begin{enumerate}
    \item A SAT-to-Ising/QUBO encoding that represents three-valued partial assignments (true/false/unassigned) natively via dual-polarity semantics.
    \item QUBO variants where implicant shrinking is expressed by polarity restrictions, and where projection penalizes only a set of variables, keeping the QUBO backbone unchanged. 
    \item A weight-regime analysis establishing conditions under which ground states correspond to a minimum-cardinality implicant, and under which satisfying assignments can be easily detected from just the final energy value. 
    \item An empirical study on random $3$-SAT and non-CNF formulas evaluating the effectiveness of the QUBO encoding to minimal and minimum implicant computation.
\end{enumerate}

\section{Preliminaries and Related Work}
\label{sec:preliminaries}



\subsection{The QUBO/Ising Optimization Problem}
\label{subsec:ising}

From an algorithmic viewpoint, quantum annealers and quantum-inspired heuristics can be regarded as solvers for quadratic optimization over binary variables. In particular, they seek low-energy configurations of an Ising Hamiltonian:
\begin{eqnarray}
\label{eq:ising_optimization}
\arg\min_{\mathbf{z} \in \{-1,1\}^{|V|}} H(\mathbf{z}),
\quad \text{where } \quad
H(\mathbf{z}) \defas
\sum_{i \in V} h_i z_i +
\sum_{(i,j) \in E} J_{i,j} z_i z_j .
\end{eqnarray}
Here each spin $z_i \in \{-1,1\}$ is a binary decision variable, and $G=(V,E)$ is an interaction graph specifying which pairs of spins may couple. The coefficients $h_i$ and $J_{i,j}$ determine the energy landscape and thus which configurations are preferred.

An equivalent representation is the \emph{quadratic unconstrained binary optimization} (QUBO) form, which minimizes a quadratic polynomial over Boolean variables $x_i \in \{0,1\}$:
\[
\arg\min_{\mathbf{x}\in\{0,1\}^{|V|}} \; \mathbf{x}^\top Q\,\mathbf{x},
\]
for some real matrix $Q$ (typically taken upper-triangular, with diagonal terms encoding linear coefficients). The two formulations are interreducible by the affine change of variables $z_i = 2x_i - 1$ (equivalently $x_i = \tfrac{1}{2}(z_i+1))$,
which converts quadratic terms $z_i z_j$ into quadratic terms $x_i x_j$ plus linear and constant offsets; constants can be dropped without affecting minimization. We therefore use Ising and QUBO interchangeably as notationally convenient (see, ~\cite{Lucas2014} for more details).

Since QUBO/Ising minimization is NP-hard, we rely on heuristic optimization algorithms that explore the energy landscape and typically return a low-energy configuration, without a guarantee of optimality. Common classical choices include simulated annealing~\cite{kirkpatrick1983optimization} and tabu/local-search variants~\cite{glover1990tabu}, while quantum annealers can be used as stochastic samplers that generate many candidate configurations~\cite{kadowaki1998quantum}.
These methods may get trapped in local minima; their effectiveness depends on how they diversify the search (e.g., via temperature schedules \cite{kirkpatrick1983optimization}, tabu memory \cite{glover1990tabu2}, or in quantum annealing via reverse annealing \cite{yamashiro2019dynamics}).

In practice, one runs the optimization algorithm for multiple trials (often over a small grid of hyperparameters), collects a pool of candidate configurations, and returns the lowest-energy configuration observed. Incremental refinement can further improve solutions: after each batch of samples, the current best configuration is used to bias or initialize the next run, a strategy shown to be effective for Ising/QUBO optimization~\cite{ding2024effective}.

\subsection{Boolean Formulas, Partial Assignments, and Projection}

Let $F$ be a Boolean formula in conjunctive normal form (CNF) over a finite set of propositional variables $\mathcal{V} = \{A_1, \dots, A_n\}$.
Let $\mathcal{L} = \{A_i, \neg A_i \mid A_i \in \mathcal{V}\}$ denote the set of all literals over $\mathcal{V}$.

A \emph{total assignment} is a function $\eta : \mathcal{V} \to \{\bot,\top\}$.
A \emph{partial assignment} is a partial function $\mu : \mathcal{V'} \to \{\bot,\top\}$, 
with $\mathcal{V'} \subseteq \mathcal{V}$. 
We write $vars(\mu) \subseteq \mathcal{V}$ to indicate the variables assigned by $\mu$.
We write $\mu_1 \supseteq \mu_2$ to denote that $\mu_1$ extends $\mu_2$, i.e., $vars(\mu_1)\supseteq vars(\mu_2)$ and $\mu_1(v)=\mu_2(v)$ for all $v\in vars(\mu_2)$.
Throughout the paper, assignments will be represented either as sets of literals, for instance $\eta = \{A_1, \neg A_2, A_3\}$, or equivalently as conjunctions of literals, e.g.\ $\eta = (A_1 \wedge \neg A_2 \wedge A_3)$.
The size $|\mu|$ is the number of literals fixed by $\mu$.

A partial assignment $\mu$ is an \emph{implicant} of $F$ if it entails $F$ (written $\mu \models F$), meaning each model of $\mu$ is a model of $F$. Equivalently, every total assignment $\eta \supseteq \mu$ satisfies $F$, i.e.,
\begin{equation*}
    \forall \eta \supseteq \mu:\quad \eta \models F.
\end{equation*}
Fix a set of \emph{visible} variables $P\subseteq\mathcal{V}$ and let $H \defas \mathcal{V}\setminus P$ be the hidden variables. A \emph{projected (partial) assignment} $\pi$ is a partial assignment with $vars(\pi)\subseteq P$. Intuitively, $\pi$ is a \emph{projected implicant} if it entails the existential projection of $F$ onto the visible variables, i.e., $\pi \models \exists H.F$. Concretely, this means that after fixing $\pi$, every completion of the remaining visible variables can still be satisfied by a suitable assignment to the hidden variables:
\begin{equation}
\label{eq:projimpl}
\forall \rho:(P\setminus vars(\pi))\to\{\bot,\top\}\ \exists \sigma:H\to\{\bot,\top\}:\quad
(\pi\wedge\rho\wedge\sigma)\models F.
\end{equation}
Let $\pi_\mu$ denote the restriction of the partial assignment $\mu$ to $P$.

\subsection{Minimal and Minimum Partial Assignments}
\label{subsec:min-vs-minimum}

An implicant is \emph{inclusion-minimal} iff no subassignment $\mu' \subset \mu$ is an implicant of $F$ (equivalently, for every literal $\ell \in \mu$, the assignment $\mu \setminus \{\ell\}$ is not an implicant of $F$). In the SAT literature, inclusion-minimal implicants are also called \emph{prime implicants} or just \emph{minimal implicants}. Minimality is a local notion: it rules out redundant literals within $\mu$, but does not compare $\mu$ to implicants using different literals. An implicant $\mu$ is \emph{minimum-cardinality} if it minimizes $|\mu|$ among all implicants of $F$. Minimum-cardinality is a global notion, and implies inclusion-minimality, but not conversely: there may exist many prime implicants of different sizes.

In the projected setting, optimality is measured only over the projected variables $P$, while the hidden variables $H=\mathcal{V}\setminus P$ serve only as auxiliary support. Thus, we measure the size of a projected implicant $\pi$ by $|\pi|$, i.e., the number of fixed literals over $P$. A projected implicant $\pi$ is \emph{projected-prime} if there is no subassignment $\pi' \subset \pi$ that is a projected implicant. A projected implicant $\pi$ is \emph{projected-minimum} if it minimizes $|\pi|$ among all projected implicants of $F$. Projected-minimum implies projected-prime, but the converse does not hold.


In the context of this paper, given a CNF formula, we want to construct a quadratic energy function whose ground states correspond to the shortest partial assignments. Thus, our aim is to characterize weight regimes that guarantee that minimum-cardinality implicants are ground states.

\subsection{Related Work}
\label{sec:related-work}

Encodings of Boolean satisfiability into Ising models and QUBO objectives have been studied for more than a decade, starting from generic penalty-based reductions from NP-complete problems to quadratic optimization~\cite{Lucas2014}. In these approaches, a CNF formula is translated into an energy function whose minimum corresponds to a satisfying assignment, while clause violations are penalized. Direct quadratic encodings for Max~2-SAT and Max~3-SAT were later proposed to avoid higher-order terms and reduce the number of auxiliary variables~\cite{chancellor2016direct}.

More recent work focused on improving SAT-to-Ising/QUBO encodings, for example, by reducing auxiliary variables, controlling coefficient ranges, or exploiting structural properties of the input formula~\cite{Zielinski2023,nusslein2023solving}. Several studies also evaluated SAT-to-Ising on real quantum annealers, analyzing penalty calibration, minor embedding, and robustness to hardware noise and limited precision~\cite{bian2017solving,bian2020solving}, as well as how formula structure and modular encoding choices affect success probability and solution quality, in particular dealing with prime factorization~\cite{ding2024effective,ding2024experimenting}.
Overall, this line of work treats SAT and MaxSAT primarily as optimization problems over quadratic energy functions, where the minimum energy corresponds to an assignment that satisfies the input problem, with no focus on computing partial models or on projection.

Minimum-cardinality implicants of a CNF formula admit a classical covering interpretation: in the plain CNF case, an implicant is a consistent set of literals that hits every clause. Hence, one could formulate the problem as a set-cover/hitting-set instance and then apply an existing QUBO/Ising encoding for set cover \cite{Lucas2014}. This, however, is not the question addressed here.
Our goal is to lift existing SAT-to-Ising encodings from total assignments to partial assignments. Standard SAT-to-Ising encodings assign one Boolean value to each variable, so every spin configuration naturally decodes as a total assignment. By contrast, our polarity-based encoding gives each Boolean variable three semantic states inside the quadratic model: true, false, and unassigned. The unassigned state is essential: it allows the optimizer to search directly over implicants, rather than over total models or over an external coverage abstraction.

This semantic shift changes the role of the energy function. The clause terms still enforce satisfiability at the formula level, while the sparsity term optimizes how many Boolean variables need to be assigned. With sufficiently large clause and consistency penalties, minimizing the energy first enforces valid partial-assignment semantics and clause satisfaction, and then minimizes cardinality among implicants. The same polarity semantics also yields natural variants: shrinking is obtained by fixing polarities that are incompatible with a given satisfying assignment, and projected objectives are obtained by penalizing only selected variables. Thus, the contribution is not a new covering characterization of implicants, but a SAT-derived Ising/QUBO semantics whose ground states are minimum-cardinality partial assignments.


\section{Quadratic Encoding of SAT with Partial Assignment Semantics}
\label{sec:encoding}

\subsection{Polarity-Splitting Representation of Variables}
\label{subsec:polarity-splitting}

Let $F$ be a CNF formula over variables $\mathcal{V}=\{A_1,\dots,A_n\}$\footnote{We assume a standard preprocessing step in which unit propagation is applied to remove unit clauses, tautological clauses are removed, and duplicate literals inside clauses are eliminated.}. For each propositional variable $A_i$, we introduce in the QUBO encoding two binary \emph{polarity spins} ($p_i$ and $n_i$) intended to encode the positive and negative literal, respectively.

This representation is inspired by the well-known \emph{dual-rail} encoding used in SAT/MaxSAT reasoning, which duplicates each variable into two rails corresponding to its two polarities \cite{bonet2018maxsat}. Concretely, we interpret each pair $(p_i,n_i)$ as a three-valued assignment:

 \begin{equation}
 \label{eq:semantic}
(p_i,n_i) =
\begin{cases}
(1,0) &\text{$A_i=\top$,} \\
(0,1) &\text{$A_i=\bot$,} \\
(0,0) &\text{$A_i$ unassigned (don't-care),} \\
(1,1) &\text{inconsistent (forbidden).}
\end{cases}
 \end{equation}

To exclude inconsistent pairs, we add the \emph{semantic consistency penalty} function:
\[
E_{\mathrm{excl}}(i) \;=\; M \cdot p_i n_i
\]
where $M>0$ is a penalty weight, high enough to guarantee assignments inconsistent with the dual rail semantic not to be valid solutions.

\subsection{Clause Encoding}
\label{subsec:clause-encoding}

Consider a clause
$
C_k = (\ell_{k,1} \lor \ell_{k,2} \lor \ell_{k,3})
$.
We now map each literal $\ell_{k,j}$ to its corresponding polarity spin, i.e. we encode the clause 
$([\![\ell_{k,1}]\!] \lor [\![\ell_{k,2}]\!] \lor [\![\ell_{k,3}]\!])$ where:
\[
[\![\ell_{k,j}]\!] =
\begin{cases}
p_i & \text{if } \ell_{k,j} = A_i,\\
n_i & \text{if } \ell_{k,j} = \neg A_i.
\end{cases}
\]

Let $a_k=[\![\ell_{k,1}]\!]$, $b_k=[\![\ell_{k,2}]\!]$, and $d_k=[\![\ell_{k,3}]\!]$.
A clause is violated exactly when all three corresponding polarity spins are $0$. A standard cubic penalty for this is $(1-a_k)(1-b_k)(1-d_k)$, but this is not quadratic. We therefore follow the approach from \cite{chancellor2016direct} and introduce one auxiliary variable $c_k$ per clause to capture the
\emph{``first-two-literals-false''} condition:
\[
c_k \;=\; (1-a_k)(1-b_k) \;\;\equiv\;\; (\neg a_k \wedge \neg b_k).
\]
Intuitively, $c_k=1$ means that neither of the first two literals satisfies the clause,
and satisfaction must be ensured by the third literal $d_k$.

We use standard 
quadratic penalty functions \cite{mandal2020compressed} enforcing $z = xy$ over $x,y,z\in\{0,1\}$:
\begin{equation}
\label{eq:lambda1}
    P_{\wedge}(x,y,z) \;=\; \lambda\bigl(xy - 2xz - 2yz + 3z\bigr) 
\end{equation}

which satisfies $P_{\wedge}(x,y,z)=0$ iff $z=xy$, and $P_{\wedge}(x,y,z)>0$ otherwise. Here $\lambda>0$ is the clause-penalty scale used to enforce the optimizer to satisfy the constraint. The higher the $\lambda$, the higher the penalty when this does not happen.
In our setting we substitute $x=1-a_k$, $y=1-b_k$, and $z=c_k$, obtaining:
\begin{equation}
P_{\wedge}(1-a_k,1-b_k,c_k)
=
\lambda\Bigl(
1 - a_k - b_k + a_k b_k
- c_k
+ 2a_k c_k
+ 2b_k c_k
\Bigr)   
\end{equation}

Finally, we penalize the unique remaining unsatisfying configuration, namely
$c_k=1$ and $d_k=0$ (which corresponds to $a_k=b_k=d_k=0$), by using the penalty function:
\[
E_{\mathrm{sat}}(k) \;=\; \lambda\, c_k(1-d_k)
\]

The final \emph{clause penalty} function for clause $C_k$ is therefore:
\begin{equation}
\label{eq:lambda2}
    E_{\mathrm{clause}}(k)
\;=\;
P_{\wedge}(1-a_k,1-b_k,c_k)
\;+\;
\lambda\,c_k(1-d_k).
\end{equation}

The $\lambda$ in Equation \ref{eq:lambda2} is the same as the one presented in Equation \ref{eq:lambda1}. By construction, $E_{\mathrm{clause}}(k)=0$ iff there exists a setting of $c_k$ such that the clause is satisfied by at least one of the polarity spins. Tables~\ref{tab:and_gadget} and~\ref{tab:sat_term} enumerate the penalty function behavior. Clauses with more than three literals are handled by iterating the same construction: auxiliary clause variables are introduced to recursively aggregate partial penalty functions $P_{\wedge}$, until the satisfaction of the whole clause is enforced by the final quadratic penalty $E_{\mathrm{sat}}(k)$.

If a clause has only two literals (i.e., after transforming a non-CNF formula via Tseitin encoding), $C_k=(\ell_{k,1}\lor \ell_{k,2})$, then we need not introduce an auxiliary variable. Let $a_k=[\![\ell_{k,1}]\!]$ and $b_k=[\![\ell_{k,2}]\!]$. The clause is violated iff $a_k=b_k=0$, hence we can enforce satisfaction with the quadratic penalty
\[
E_{\mathrm{clause}}(k) \;=\; \lambda\,(1-a_k)(1-b_k),
\]
which is already quadratic and equals $0$ iff $a_k\lor b_k$ is satisfied by at least one polarity spin.

\begin{table}[t]
\centering

\begin{minipage}{0.48\linewidth}
\centering
\begin{tabular}{ccc|c}
$a_k$ & $b_k$ & $c_k$ & $P_{\wedge}(1-a_k,1-b_k,c_k)$ \\
\hline
0 & 0 & 0 & $\lambda$ \\
0 & 0 & 1 & 0 \\
0 & 1 & 0 & 0 \\
0 & 1 & 1 & $\lambda$ \\
1 & 0 & 0 & 0 \\
1 & 0 & 1 & $\lambda$ \\
1 & 1 & 0 & 0 \\
1 & 1 & 1 & $3\lambda$ \\
\end{tabular}
\caption{Truth table for the AND-penalty enforcing
$c_k=(1-a_k)(1-b_k)$.}
\label{tab:and_gadget}
\end{minipage}
\hfill
\begin{minipage}{0.48\linewidth}
\centering
\begin{tabular}{cc|c}
$c_k$ & $d_k$ & $E_{\mathrm{sat}}(k)=\lambda c_k(1-d_k)$ \\
\hline
0 & 0 & 0 \\
0 & 1 & 0 \\
1 & 0 & $\lambda$ \\
1 & 1 & 0 \\
\end{tabular}
\caption{Truth table for $E_{\mathrm{sat}}(k)$:
it penalizes exactly the case $c_k=1$ and $d_k=0$ (i.e. all three variables in clause $k$ are false).}
\label{tab:sat_term}
\end{minipage}

\end{table}

\subsection{Sparsity-Promoting Term}
\label{subsec:sparsity}

To bias solutions toward short partial assignments, we add a \emph{sparsity penalty} function
\[
E_{\mathrm{sparse}}
\;=\;
\gamma \sum_{i=1}^{n}(p_i+n_i),
\]
which penalizes active spins (i.e., polarity spins set to 1). Here $\gamma>0$ is the sparsity weight, i.e., the cost of activating a polarity spin. Under proper weight regimes (discussed in more detail in Section~\ref{sec:weights}), clause satisfaction and polarity consistency dominate sparsity, so the optimizer prefers setting $(p_i,n_i)=(0,0)$ whenever this does not compromise satisfaction. This term is also the mechanism through which we later realize shrinking and projection: by applying sparsity weights only to a designated visible set $P\subseteq\mathcal V$, hidden variables may act as unpenalized support for satisfiability.

\subsection{The SAT-to-Ising/QUBO Objective}
\label{subsec:full-objective}

Given $F$ with clauses $C_1,\dots,C_m$ over $n$ variables, the full quadratic energy expression is:

\begin{equation}
\label{eq:EF}
    E_F
=
\sum_{i=1}^n M\, p_i n_i
+
\sum_{k=1}^m E_{\mathrm{clause}}(k)
+
\gamma \sum_{i=1}^n (p_i + n_i)
\end{equation}

The resulting quadratic model uses, in the worst case, $O\!\left(n + \sum_{k=1}^m |C_k|\right)$ binary variables: two polarity spins for each Boolean variable and then auxiliary spins for each $P_{\wedge}$.

\section{Weight Regimes and Minimum-Cardinality Partial Assignments}
\label{sec:weights}

The intended goal is that ground states correspond to satisfying assignments
that activate a minimum number of polarity spins. The QUBO weights must be tuned to avoid ground states corresponding to assignments violating the semantics from Section \ref{sec:encoding}.


\begin{theorem}[Minimum-cardinality satisfying partial assignments]
\label{prop:minimal-partial-models}
Let $F$ be a satisfiable CNF formula and assume the weights satisfy
\[
\lambda > n\,\gamma
\qquad\text{and}\qquad
M > n\,\gamma.
\]
Then every ground state of $E_F$ is ($i$) consistent, ($ii$) satisfies all clauses of $F$,
and ($iii$) minimizes the number of spin polarities set to 1 among all consistent satisfying assignments, thus corresponding to a minimum-cardinality implicant for $F$. Thus, the encoding reduces the computation of minimum-cardinality implicants to exact ground-state optimization.
\end{theorem}

\begin{proof}
Let $E_F$ be the energy function based on the double polarity spin encoding in Section~\ref{sec:encoding}.
Let $\mathbf{z}^\star$ be a ground state (minimum-energy spin configuration) of $E_F$, and
let $\eta^\star$ be the induced three-valued assignment over $\mathcal V$
obtained from $\mathbf{z}^\star$ by reverting Equation \ref{eq:semantic}.

Since $F$ is satisfiable, fix an arbitrary total satisfying assignment
$\tilde{\eta}$ of $F$, and let $\tilde{\mathbf{z}}$ be its consistent polarity spins
encoding. For this configuration, all clause penalties can be set to 0, hence $E_F(\tilde{\mathbf{z}})= n\gamma$.

\textbf{(i) Clause satisfaction.}
Assume for contradiction that $\eta^\star$ violates at least one clause of $F$.
Consider the minimum value of the clause penalty function for a violated clause
(i.e., $a_k=b_k=d_k=0$). By construction of $E_{\mathrm{clause}}(k)$, after
minimizing over $c_k$ the clause contributes at least $\lambda$.
Therefore
\[
E_F(\mathbf{z}^\star)\;\ge\;\lambda.
\]
Since $\lambda>n\gamma$ and $E_F(\tilde{\mathbf{z}})\le n\gamma$, we obtain
$E_F(\mathbf{z}^\star)>\,E_F(\tilde{\mathbf{z}})$, contradicting that
$\mathbf{z}^\star$ is a ground state. Hence every ground state satisfies all
clauses.

\textbf{(ii) Polarity consistency.}
Assume for contradiction that $\mathbf{z}^\star$ is inconsistent, i.e., there
exists a variable $A_i$ with $p_i=n_i=1$. Then the semantic consistency term contributes $M$, so
\[
E_F(\mathbf{z}^\star)\;\ge\;M.
\]
Using $M>n\gamma$ and again $E_F(\tilde{\mathbf{z}})\le n\gamma$, we derive
$E_F(\mathbf{z}^\star)>\,E_F(\tilde{\mathbf{z}})$, contradicting that $\mathbf{z}^\star$ is a ground state. Therefore every ground
state is polarity-consistent.

\textbf{(iii) Minimum-cardinality among satisfying consistent assignments.}
By (i) and (ii), $\mathbf{z}^\star$ is consistent and satisfies all clauses, so
all clause penalties and semantic consistency penalties vanish and
\[
E_F(\mathbf{z}^\star)=\gamma\cdot s(\mathbf{z}^\star),
\]
where $s(\mathbf{z})$ is the number of active polarity spins in
$\mathbf{z}$.
Assume for contradiction that there exists another satisfying
configuration $\mathbf{z}'$ with strictly fewer active polarities:
$s(\mathbf{z}')<s(\mathbf{z}^\star)$.
Then $E_F(\mathbf{z}')=\gamma s(\mathbf{z}')<\gamma s(\mathbf{z}^\star)=E_F(\mathbf{z}^\star)$,
contradicting that $\mathbf{z}^\star$ is a ground state.

Thus $\mathbf{z}^\star$ minimizes the number of active polarities among all
consistent satisfying configurations, i.e., it corresponds to a
minimum-cardinality satisfying partial assignment (and hence, under our CNF
notion, to a minimum-cardinality implicant).
\end{proof}

\begin{corollary}
\label{cor:low-energy-feasible}
Assume $\lambda > n\gamma$ and $M > n\gamma$.   Then any spin configuration $\mathbf{z}$ with $E_F(\mathbf{z}) \le n\gamma$ is polarity-consistent and satisfies all clauses of $F$ (hence induces a consistent satisfying three-valued assignment under the polarity semantics). 
\end{corollary}

\begin{proof}
If $\mathbf{z}$ violates some clause, then the corresponding clause term contributes at least $\lambda$, so $E_F(\mathbf{z}) \ge \lambda > n\gamma$, a contradiction. If $\mathbf{z}$ is inconsistent, then the semantic consistency term contributes $M$, so $E_F(\mathbf{z}) \ge M > n\gamma$, a contradiction.
\end{proof}

A concrete admissible choice satisfying the above condition is $\gamma = 1$ and $\lambda = M = n + 1$. These coefficients are integer-valued, have magnitude $O(n)$, and can be rescaled uniformly to meet hardware-dependent bounds without affecting the ground-state semantics. By \cref{cor:low-energy-feasible}, under this weight regime any configuration with energy at most $n\gamma$ must be clause-satisfying and polarity-consistent.

\begin{example}
We illustrate the encoding of the one-clause formula
$F \;=\; (x_1 \lor \neg x_2 \lor x_3)$.
For this clause we have $a_1 = p_1$, $b_1 = n_2$, and $d_1 = p_3$.
Introducing the auxiliary spin $c_1$, the QUBO energy (for this single-clause instance) is
\[
\begin{aligned}
E_F
&=
M(p_1 n_1 + p_2 n_2 + p_3 n_3)
\;+\;
\gamma(p_1 + n_1 + p_2 + n_2 + p_3 + n_3)
\\[2pt]
&\quad
+\;
\lambda(
1 - p_1 - n_2 + p_1 n_2
- c_1
+ 2 p_1 c_1
+ 2 n_2 c_1
)
\;+\;
\lambda\, c_1(1 - p_3).
\end{aligned}
\]
Throughout the example we assume the weight regime $\lambda>n\gamma$ and $M>n\gamma$
from Section~\ref{sec:weights}, where here $n=3$.

\medskip
\noindent
\textbf{Case 1: inconsistent polarity.}
Consider the contradictory polarity assignment $\mathbf{z_1}$: 
\[
p_1 = n_1 = 1,
\qquad
p_2 = n_2 = p_3 = n_3 = 0.
\]
Since $p_1=1$, the clause is already satisfied regardless of $(n_2,p_3)$, hence the Ising model can choose $c_1=0$ to make the clause penalties vanish. However, the semantic consistency term contributes $M$ and sparsity contributes $2\gamma$, so
\[
E_F(\mathbf{z_1}) = M + 2\gamma.
\]
Because $M>3\gamma$, this energy exceeds the sound bound $3\gamma$ for clause satisfaction and consistency.

\medskip
\noindent
\textbf{Case 2: all polarities inactive.}
Let $\mathbf{z_2}$ set all polarity spins to zero:
\[
p_1=n_1=p_2=n_2=p_3=n_3=0.
\]
Then the clause $(x_1 \lor \neg x_2 \lor x_3)$ is violated. Indeed, setting $p_1=n_2=p_3=0$ forces all three literals to be false. Minimizing over $c_1$ yields a clause-violation penalty of at least $\lambda$: choosing $c_1=1$ satisfies the penalty function $c_1=(1-p_1)(1-n_2)$, but then the term $\lambda c_1(1-p_3)$ contributes $\lambda$.
Thus:
\[
E_F(\mathbf{z_2}) = \lambda,
\]
which is strictly larger than $3\gamma$ by the assumption $\lambda>3\gamma$.
Therefore, the configuration is not a satisfying assignment ground state.


\medskip
\noindent
\textbf{Case 3: a total satisfying assignment.}
Let $\eta_1 = x_1 \wedge \neg x_2 \wedge x_3$, mapped to the polarity spin configuration $\mathbf{z_3}$ with polarity encoding:
\[
(p_1,n_1)=(1,0),\qquad (p_2,n_2)=(0,1),\qquad (p_3,n_3)=(1,0).
\]
The clause is satisfied because $p_1=1$, so we get zero
clause penalty. The configuration is consistent, and sparsity activates three polarity spins, hence:
\[
E_F(\mathbf{z_3}) = 3\gamma.
\]
This matches the certificate bound $n\gamma$ for $n=3$, so it is consistent
with clause satisfaction under the weight regime.

\medskip
\noindent
\textbf{Case 4: a partial satisfying assignment.}
Finally, consider the partial assignment $\eta_2 = x_1$, setting the polarity spin configuration $\mathbf{z_4}$ to:
\[
(p_1,n_1)=(1,0),\qquad (p_2,n_2)=(0,0),\qquad (p_3,n_3)=(0,0).
\]
The clause is satisfied because $p_1=1$, so again choosing $c_1=0$ sets the clause penalty to zero.
The configuration is consistent, and sparsity activates only one polarity, hence:
\[
E_F(\mathbf{z_4}) = \gamma.
\]
Therefore, $\eta_2$ has strictly smaller energy than any total satisfying assignment for this instance,
and is preferred by the sparsity term while still satisfying $F$.
\end{example}

\section{Implicant Shrinking via Fixed Polarities}
\label{sec:implicant-shrinking}

Section \ref{sec:encoding} considered minimizing $E_F$ over the full search space to obtain minimum-cardinality implicants from scratch. In many workflows, however, a total satisfying assignment $\eta$ is already available (e.g., produced by a SAT/AllSAT solver), and the goal is to shrink $\eta$ into a shorter implicant by deleting literals that are not required.

Let $\eta$ be a total satisfying assignment of $F$, written as a conjunction of
literals
\[
\eta \;=\; \bigwedge_{i=1}^n \ell_i,
\qquad\text{where }\ell_i\in\{A_i,\neg A_i\}.
\]
We restrict the value of a subset of polarity spins so that we may only drop literals from $\eta$:
\begin{equation}
\label{eq:restrict}
    \begin{cases}
\eta(A_i)=\top: & n_i = 0,\quad p_i\in\{0,1\}\ \text{free},\\[2pt]
\eta(A_i)=\bot: & p_i = 0,\quad n_i\in\{0,1\}\ \text{free}.
\end{cases}
\end{equation}
Thus, for each variable $A_i$, the polarity spin consistent with $\eta$ remains switchable
(kept if set to $1$, dropped if set to $0$), while the opposite polarity is
forbidden to get to 1.

Under these constraints, every semantically valid polarity spin configuration decodes to a partial assignment $\mu\subseteq\eta$ obtained by deleting some literals of $\eta$. In particular, setting the remaining allowed polarity spin to $0$ makes $A_i$ unassigned, which matches the implicant semantics: an unassigned variable must be a genuine don't-care, i.e., the formula must remain satisfied for both truth values of that variable.

\begin{theorem}[Minimum-cardinality implicants within a given assignment]
\label{lem:implicant-shrinking}
Assume $F$ is satisfiable and the same weights configuration as in Theorem \ref{prop:minimal-partial-models}. Let $\eta$ be a total satisfying assignment of $F$, and let $E_F^\eta$ be the
energy obtained from $E_F$ by imposing the polarity restrictions from Equation \ref{eq:restrict}.
Then every ground state of $E_F^\eta$ decodes to a partial assignment
$\mu\subseteq\eta$ that is an implicant of $F$ and has minimum cardinality among
all implicants contained in $\eta$. Consequently, $\mu$ is also inclusion-minimal  within the family $\{\mu'\subseteq\eta \mid \mu'\text{ implicant of }F\}$.
\end{theorem}

\begin{proof}
Let $\mathbf z^\star$ be a ground state of $E_F^\eta$ and let $\mu^\star\subseteq\eta$
be its decoded partial assignment.

\textbf{(1) Clause satisfaction.}
Assume for contradiction that $\mu^\star$ does not satisfy some clause anymore (i.e., removing one literal that was necessary for the satisfiability of a clause). By the construction of $E_{\mathrm{clause}}(k)$, after minimizing over the auxiliary
variable $c_k$ a violated clause contributes at least $\lambda$, hence
$E_F^\eta(\mathbf z^\star)\ge \lambda$. On the other hand, the polarity encoding
$\tilde{\mathbf z}$ of the original total model $\eta$ is feasible for the
restricted problem and satisfies all clauses, so its energy is purely sparsity:
$E_F^\eta(\tilde{\mathbf z})=n\gamma$. Since $\lambda>n\gamma$, we obtain
$E_F^\eta(\mathbf z^\star)>E_F^\eta(\tilde{\mathbf z})$, contradicting that
$\mathbf z^\star$ is a ground state. Hence $\mu^\star$ is clause-satisfying.

\textbf{(2) Polarity consistency.}
In the reduced problem $E_F^\eta$, for each variable $A_i$ one polarity is
hard-fixed to $0$ (namely, $n_i=0$ if $\eta(A_i)=\top$ and $p_i=0$ if
$\eta(A_i)=\bot$). Thus, for every feasible configuration we have
$p_i n_i = 0$ for all $i$, i.e., feasibility already enforces polarity
consistency. Thus, the semantic consistency penalties are always $0$ on the
feasible set.

\textbf{(3) Minimum cardinality within $\eta$.}
By (1) and (2), all clause and semantic consistency penalties vanish on $\mathbf z^\star$,
so $E_F^\eta(\mathbf z^\star)=\gamma\cdot s(\mathbf{z^*})$, i.e., energy is exactly
$\gamma$ times the number of kept literals from $\eta$.
Assume for contradiction that there exists another implicant $\mu'\subseteq\eta$
with $|\mu'|<|\mu^\star|$. Its polarity encoding is feasible for the restricted
problem and has energy $E_F^\eta(\mu')=\gamma|\mu'|<\gamma|\mu^\star|=E_F^\eta(\mathbf z^\star)$,
contradicting that $\mathbf z^\star$ is a ground state. Therefore, $\mu^\star$ has
minimum cardinality among implicants consistent with $\eta$.
\end{proof}

\begin{example}
Consider the one-clause formula
\[
F = (x_1 \lor \neg x_2 \lor x_3),
\qquad
\eta = x_1 \wedge \neg x_2 \wedge x_3 .
\]
Translated to a polarity spin encoding, $\eta$ corresponds to
\[
(p_1,n_1)=(1,0),\qquad (p_2,n_2)=(0,1),\qquad (p_3,n_3)=(1,0).
\]

To restrict optimization to partial assignments $\mu\subseteq \eta$, we forbid the polarity spins opposite to $\eta$ for each variable:
\[
n_1=0,\qquad p_2=0,\qquad n_3=0,
\]
leaving the compatible polarities $p_1,n_2,p_3$ free in $\{0,1\}$.

Since one polarity spin per variable is hard-fixed to $0$, every feasible configuration satisfies $p_i n_i=0$, hence the semantic consistency term vanishes on the restricted search space. Moreover, for the clause $C=(x_1 \lor \neg x_2 \lor x_3)$ we have $a_1=p_1$, $b_1=n_2$, $d_1=p_3$. Minimizing the clause penalty function over $c_1$ yields:
\[
\min_{c_1\in\{0,1\}} E_{\mathrm{clause}}(1)
=
\begin{cases}
0 & \text{if } (p_1 \lor n_2 \lor p_3)=1,\\
\lambda & \text{if } (p_1,n_2,p_3)=(0,0,0).
\end{cases}
\]

Therefore, among satisfying restricted configurations, the objective reduces to the term
\[
E^\eta_F = \gamma(p_1+n_2+p_3).
\]
The satisfying configurations and their energies are thus:
\[
\begin{array}{ccl}
(p_1,n_2,p_3)=(1,1,1) &\leftrightarrow& x_1\wedge \neg x_2\wedge x_3,\quad E^\eta_F=3\gamma,\\[2pt]
(p_1,n_2,p_3)\in\{(1,1,0),(1,0,1),(0,1,1)\} &\leftrightarrow& \text{two-literal implicants},\quad E^\eta_F=2\gamma,\\[2pt]
(p_1,n_2,p_3)\in\{(1,0,0),(0,1,0),(0,0,1)\} &\leftrightarrow& \text{single-literal implicants},\quad E^\eta_F=\gamma.
\end{array}
\]
The configuration $(0,0,0)$ violates the clause and has energy
at least $\lambda>3\gamma$, hence it cannot be optimal. Every ground state has energy $\gamma$ and corresponds to a
single-literal implicant contained in $\eta$, e.g.\ $(p_1,n_2,p_3)=(1,0,0)$ decodes
to $\mu=x_1$.
\end{example}

\section{Projected Models and Projected Implicant Shrinking}
\label{sec:projected}

In many reasoning tasks, only a subset of variables is of interest. This situation arises, for instance, when non-CNF formulas are converted into CNF using Tseitin-style transformations~\cite{tseitin1983complexity} or Plaisted--Greenbaum encoding~\cite{plaisted1986structure}: fresh variables introduced by the transformation are existentially quantified and should not appear in the final models. More generally, explanation and optimization tasks often focus on a set of \emph{visible} variables, with other variables acting as auxiliary support \cite{Moehle2025}.




To bias solutions toward small assignments over $P$, we keep the clause and semantic consistency penalties unchanged, while sparsity is restricted to variables in $P$ only. We define the \emph{projected sparsity} term and update the final penalty function correspondingly:
\[
E_{\mathrm{proj}}
\;=\;
\gamma \sum_{A_i\in P}(p_i+n_i); \qquad\qquad\qquad E^{\mathrm{proj}}_F
\;=\;
\sum_{i=1}^n M\,p_i n_i
\;+\;
\sum_{k=1}^m E_{\mathrm{clause}}(k)
\;+\;
E_{\mathrm{proj}}
\]

Intuitively, hidden variables $H$ may freely take whichever polarity configuration satisfies the clauses, without paying sparsity. Compared to Equation \ref{eq:EF}, only variables in $P$ are pushed toward the unassigned state $(p_i,n_i)=(0,0)$ unless they are necessary for clause satisfaction.

We can use the encoding to perform implicant shrinking over a set of important variables $P$. In our encoding, this is obtained without changing the clause or consistency terms: we keep the objective $E^{\mathrm{proj}}_F$, but fix the polarity spins of the visible variables so that each $A_i\in P$ may either keep the value prescribed by $\eta$ or become unassigned, but can never flip. 
Concretely, for each $A_i\in P$ we impose:
\[
n_i=0\ \ \text{if}\ \ \eta(A_i)=\top, 
\qquad
p_i=0\ \ \text{if}\ \ \eta(A_i)=\bot,
\]
Hidden variables $H$ remain unconstrained and can adapt to satisfy the clause penalty function without paying sparsity penalty. 

A natural question is whether the projected sparsity encoding guarantees a projected implicant over $P$ that is minimum (projected minimum-cardinality) or even minimal (projected-prime). In general, it does not.

\begin{example}
\label{ex:xor-separation}
Let us consider the formula
\[
F \;=\; (x\lor y)\ \wedge\ (\neg x\lor \neg y),
\qquad
P=\{x\},\quad H=\{y\}.
\]
Under the projected-implicant semantics from Section~\ref{sec:preliminaries}, the empty projected assignment is a valid projected implicant: for $\rho(x)=\top$ choose
$\sigma(y)=\bot$, and for $\rho(x)=\bot$ choose $\sigma(y)=\top$.
Thus $\pi=\emptyset$ is projected-minimum (it fixes no variable in $P$).

However, our encoding cannot produce a decoded partial assignment $\mu$ that leaves $x$ unassigned. If $x$ is left free, then both completions $x=\top$ and $x=\bot$ must be handled, but in $F$ this forces $y$ to change with $x$. Since $\mu$ is decoded from a single configuration, it cannot accommodate both cases at once.
\end{example}

Recall the standard projected-implicant semantics definition, where $\pi$ is an implicant if:
\[
\forall \rho:(P\setminus vars(\pi))\to\{\bot,\top\}\ \exists \sigma:H\to\{\bot,\top\}:\ (\pi\wedge\rho\wedge\sigma)\models F.
\]
In contrast, the assignment $\mu$ returned by the QUBO encoding minimization enforces a uniform witness:
\[
\exists \sigma^\star: H\to\{\bot,\top\}\ \text{s.t.}\ 
\forall \rho: (P\setminus vars(\pi_\mu))\to\{\bot,\top\}:\quad
(\pi_\mu \wedge \rho \wedge \sigma^\star)\models F.
\]
which is a different (and strictly stronger) requirement: it must hold for all completions at once, rather than allowing a hidden completion to adapt to the visible completion. Example~\ref{ex:xor-separation} highlights the source of the mismatch: the value of the hidden variable must change with the completion of the visible one, so no single hidden assignment can witness all visible completions at once when these circular dependencies are in the CNF formulation.

This worst-case limitation does not preclude much stronger behavior in practice, in particular in real-world non-CNF instances. For non-CNF formulas converted into CNF by a Negative Normal Form (NNF) + Plaisted-Greenbaum pipeline, it is known that for every short partial assignment $\mu$ over the original variables, if $\mu$ entails the original formula, then there exists an assignment $\sigma$ to the introduced auxiliary variables such that $\mu \wedge \sigma$ entails the CNF encoding~\cite{masina2025cnf}.
Under this condition, the projected case reduces to the ordinary implicant case: every projected implicant $\pi$ generated from the QUBO model admits a single hidden assignment $\sigma$ such that $\pi \wedge \sigma$ entails $F$. Hence, the minimum-cardinality analysis extends to projected assignments as well; the formal statement is given in Appendix \ref{appendix:theorem}.

\section{Experimental Evaluation}
\label{sec:experiments}


Our experimental goal is the empirical validation of the optimization semantics induced by the QUBO model\footnote{A direct runtime comparison against specialized SAT-side implicant procedures would primarily reflect backend engineering choices on the QUBO optimizer, which are orthogonal to the scope of this paper.}. Theorems \ref{prop:minimal-partial-models} and \ref{lem:implicant-shrinking} characterize exact ground states; the experiments assess whether low-energy configurations returned by a generic QUBO optimizer decode to valid implicants and how often they satisfy minimality and minimum-cardinality in practice.


\subsection{Benchmarks}
\label{sec:exp:benchmarks}

The first benchmark consists of satisfiable random 3-SAT instances. We generate $10$ instances with $n$ variables ($n \in \{8,12,16,\dots,200\}$) at fixed clauses-to-variables density $m/n = 1.5$.
The structure of the solution space is crucial: near the SAT phase transition,
at a clause-to-variable ratio around $4.2$, random $3$-CNF formulae typically
admit very few satisfying assignments. Thus, it is unlikely that a
short partial assignment can cover many total satisfying assignments, since
there are few models to compact in the first place. In contrast, lower-density regimes, such as $m/n = 1.5$, admit many satisfying assignments while remaining non-trivial~\cite{spallitta2024disjoint}, and are therefore better suited for our task. 

We also consider 110 non-CNF formulas from recent work~\cite{masina2023cnf,masina2025cnf}, with number of variables $n \in \{20, 24, 28,\dots,60\}$. These instances are Boolean
formulas obtained by nesting Boolean operators up to a fixed depth, thus not directly written in CNF. We convert them to CNF using the pipeline in Masina et al.~\cite{masina2023cnf,masina2025cnf}: first NNF, then apply the Plaisted--Greenbaum transformation \cite{plaisted1986structure}. The conversion introduces auxiliary variables, which are treated as the hidden set $H$.

The two benchmark families are chosen to test complementary properties. Random 3-SAT isolates the core minimum-cardinality implicant behavior. The non-CNF benchmarks create a realistic projected setting in which visible and hidden variables must be treated differently.

\subsection{Experimental setup}
\label{sec:exp:setup}

For each instance, we construct the QUBO objective~$E_F$ as described in Section~\ref{sec:encoding}. We then optimize it using simulated annealing (SA) as a baseline quadratic optimizer, drawing $1000$ samples per run and returning the lowest-energy configuration observed. Decoding the polarity spins via Equation  \ref{eq:semantic} returns a partial assignment.


In addition to this one-shot SA run approach (\textsc{Basic}), we also test an iterative refinement variant (\textsc{Iter}), inspired by Ding et al.~\cite{ding2024effective}.
Starting from the best sample returned by SA, we freeze its polarity pattern and re-run SA on the resulting restricted model, pushing additional variables into the don't-care state.
In iterative rounds, we use $100$ samples and stop when the assignment no longer shrinks. For people interested in better understanding \textsc{Iter}, we provide a guided example in Appendix \ref{app:iterative-shrinking-example}.

We report results under two complementary design choices.
First, we consider whether optimization is performed directly on the full QUBO encoding from Section \ref{sec:encoding} (\textsc{Full}) or used to \emph{shrink} a pre-computed total satisfying assignment as in Section \ref{sec:implicant-shrinking} (\textsc{Shrink}).
Second, we consider whether sparsity is penalized over all variables as in Sections \ref{sec:encoding}-\ref{sec:implicant-shrinking} (\textsc{Standard}) or only over a designated visible set $P$ as in Section \ref{sec:projected} (\textsc{Projected}).

\subsection{Experiment 1: Compactness of partial assignments}
\label{sec:exp:compactness}

We first measure the compactness of our encoding (i.e., the length of the short assignments produced by \textsc{Basic}) on the random 3-SAT benchmarks.
In standard mode, we report the ratio $|\mu|/n$, whereas in projected mode, we randomly preselect the set $P$ (with $|P|$ being 75\% of the total variables) and report the ratio $|\pi|/|P|$.
Figure~\ref{fig:supportratio-all} reports mean$\pm$std over $10$ instances per $n$, while Figures~\ref{fig:scatter-standard} and~\ref{fig:scatter-projected} report the retrieved size per instance via scatterplots.

At fixed density, the ratios are approximately constant with respect to $n$. In standard mode, \textsc{Full} returns smaller assignments than \textsc{Shrink}, as expected: in \textsc{Shrink} the polarity pattern is fixed by the initial SAT model, so the optimizer can only delete literals under that sign pattern, reducing the degrees of freedom for reaching more compact assignments.
The same qualitative trend holds in projected mode.

\begin{figure*}
    \centering
\begin{minipage}{0.4\textwidth}
    \centering
    \includegraphics[width=\linewidth]{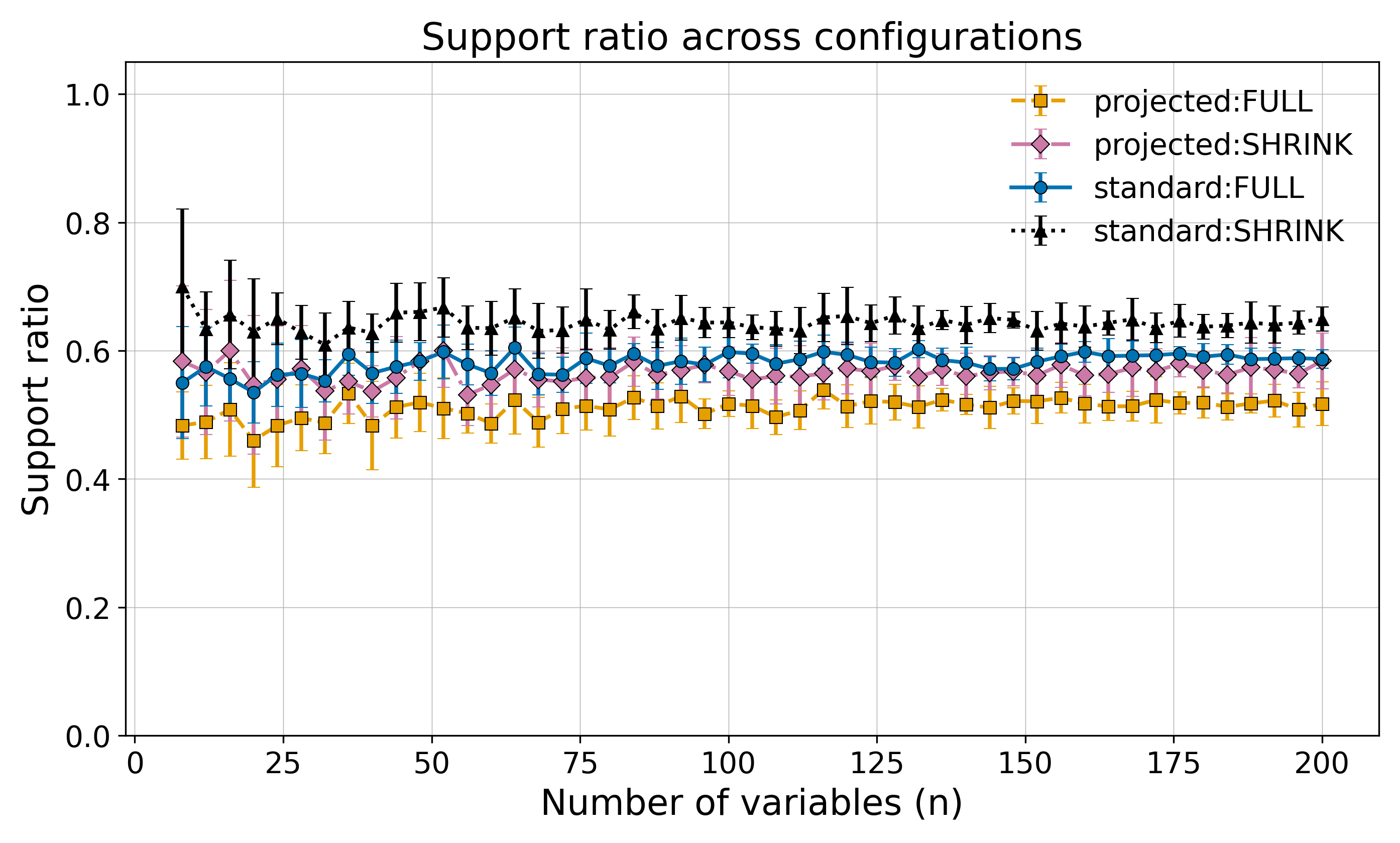}
    \caption{Ratio of assigned variables across configurations (standard: $|\mu|/n$, projected: $|\pi|/|P|$). Mean $\pm$ std over 10 instances per $n$.}
    \label{fig:supportratio-all}
  \end{minipage}
  \begin{minipage}{0.28\textwidth}
    \centering
    \includegraphics[width=1\linewidth]{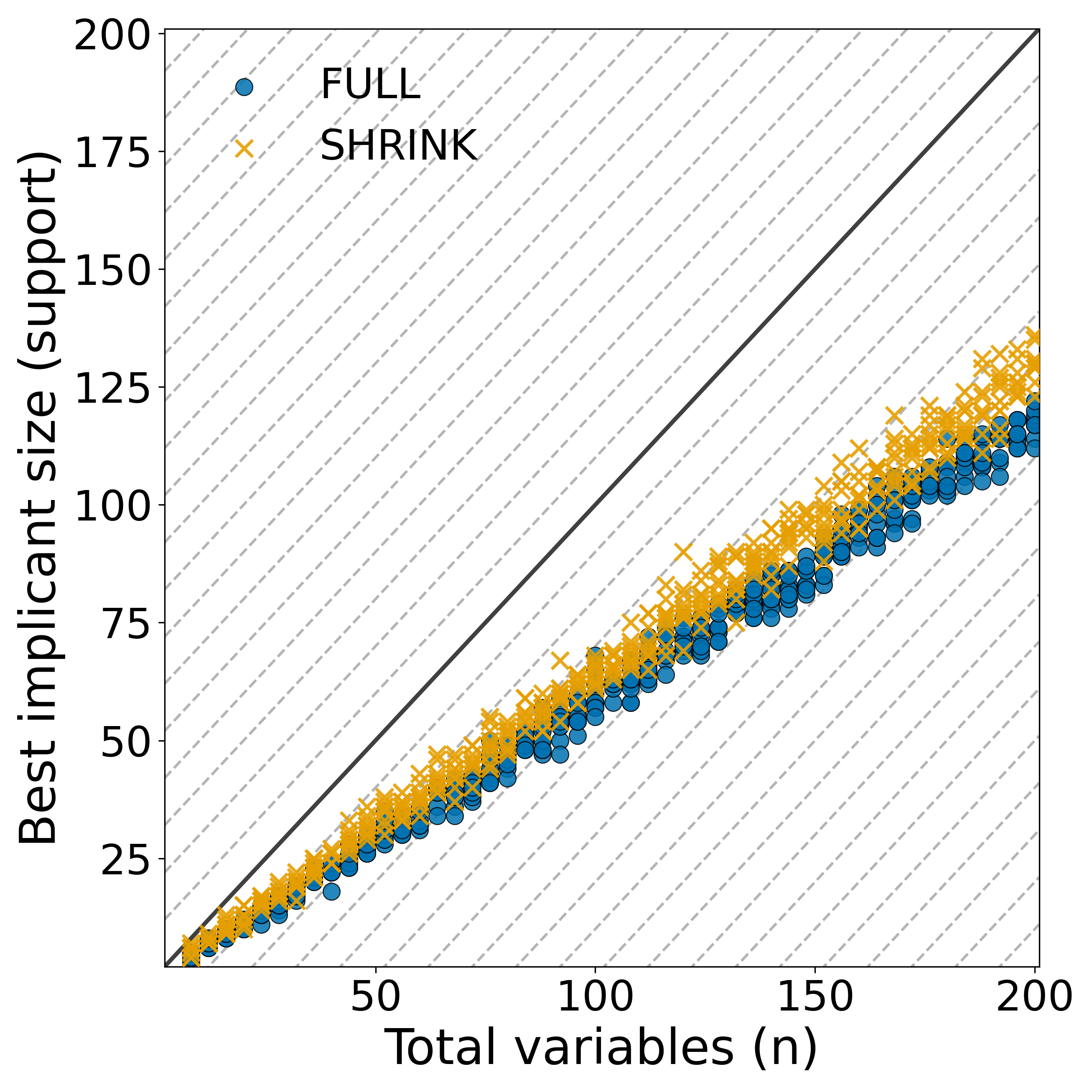}
    \caption{Standard mode: best implicant size $|\mu|$ vs $n$.}
    \label{fig:scatter-standard}
  \end{minipage}
  \hfill
  \begin{minipage}{0.28\textwidth}
    \centering
    \includegraphics[width=\linewidth]{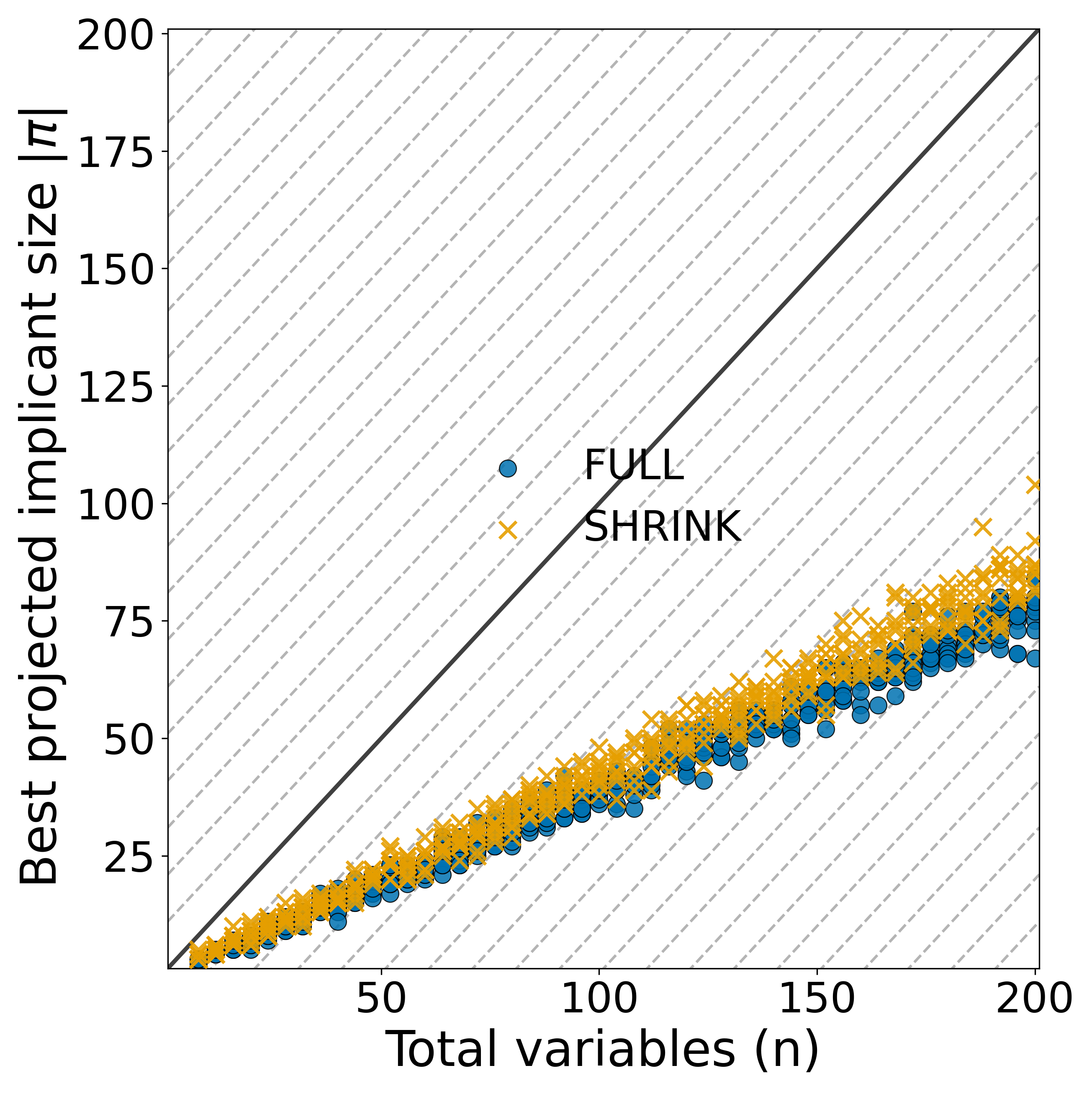}
    \caption{Projected mode: best projected size $|\pi|$ vs $n$. }
    \label{fig:scatter-projected}
  \end{minipage}
\end{figure*}

\subsection{Experiment 2: Certification of minimality and minimum-cardinality}
\label{sec:exp:optimality}

\begin{figure*}[t]
  \centering
  \begin{minipage}{1\textwidth}
    \centering
    \includegraphics[width=0.85\linewidth]{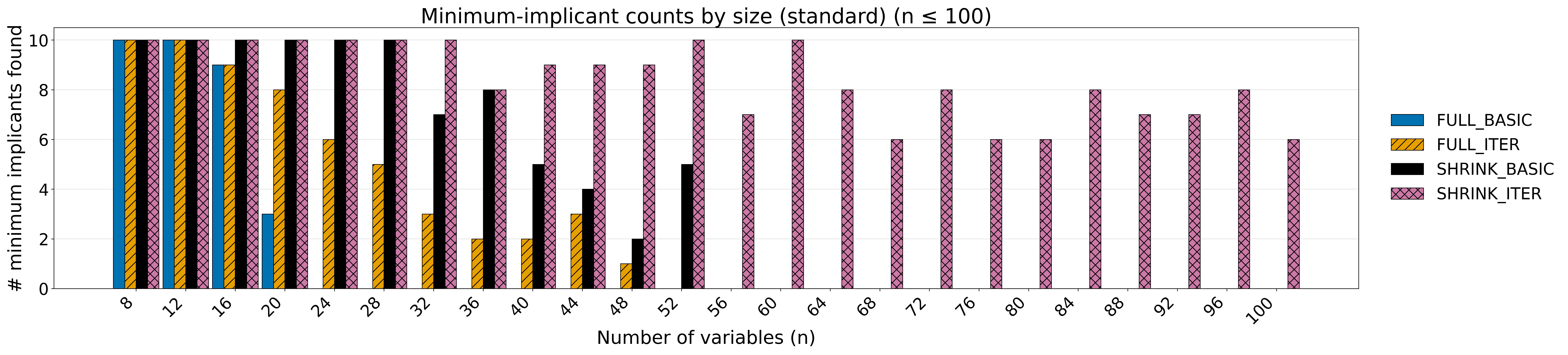}
    \label{fig:min-standard}
  \end{minipage}
  \hfill
  \begin{minipage}{1\textwidth}
    \centering
    \includegraphics[width=0.85\linewidth]{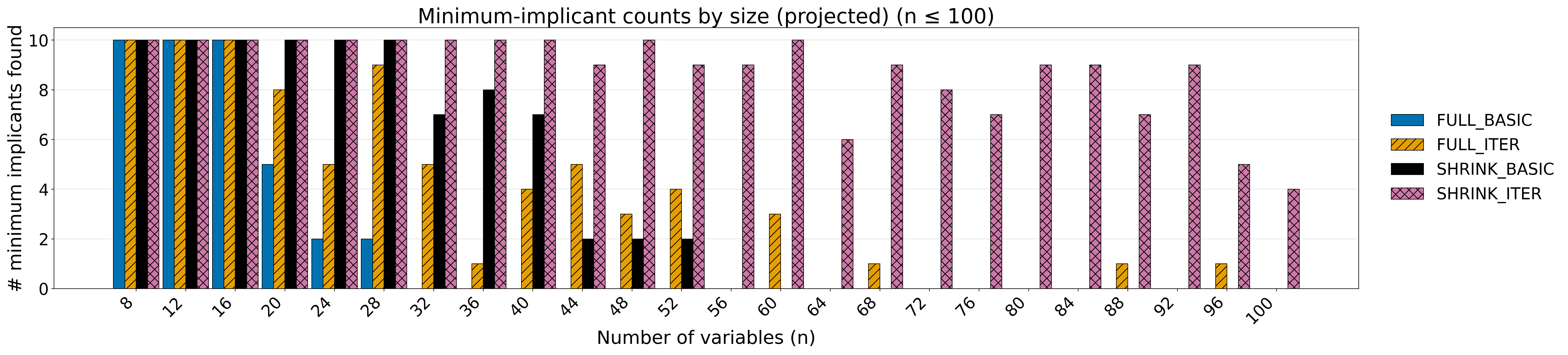}
    \caption{Minimum implicants found for $n\le 100$.
  \textbf{Top}: standard. \textbf{Bottom}: projected.}
    \label{fig:min-projected}
  \end{minipage}
\end{figure*}

\begin{figure*}[t]
  \centering
  \begin{minipage}{1\textwidth}
    \centering
    \includegraphics[width=0.85\linewidth]{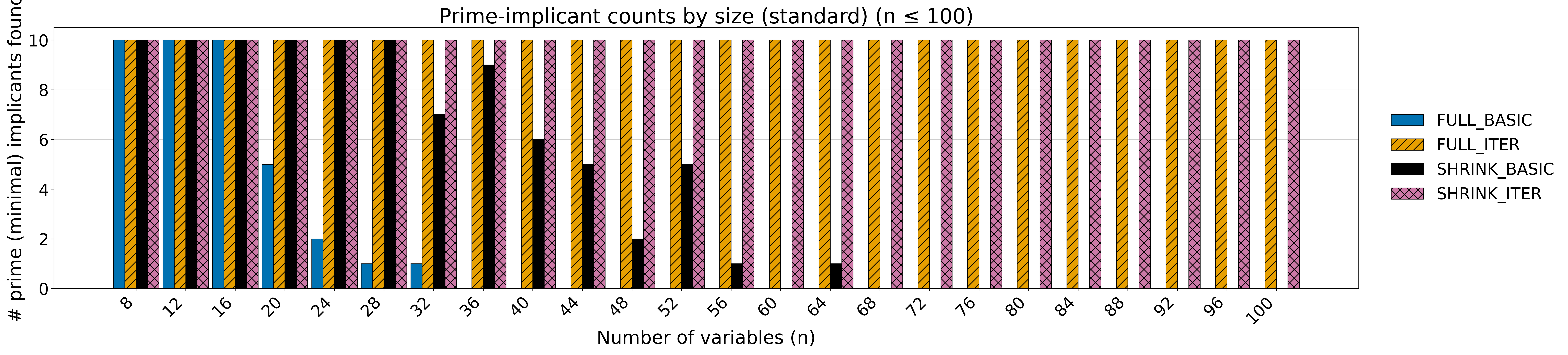}
    \label{fig:prime-standard}
  \end{minipage}
  \hfill
  \begin{minipage}{1\textwidth}
    \centering
    \includegraphics[width=0.85\linewidth]{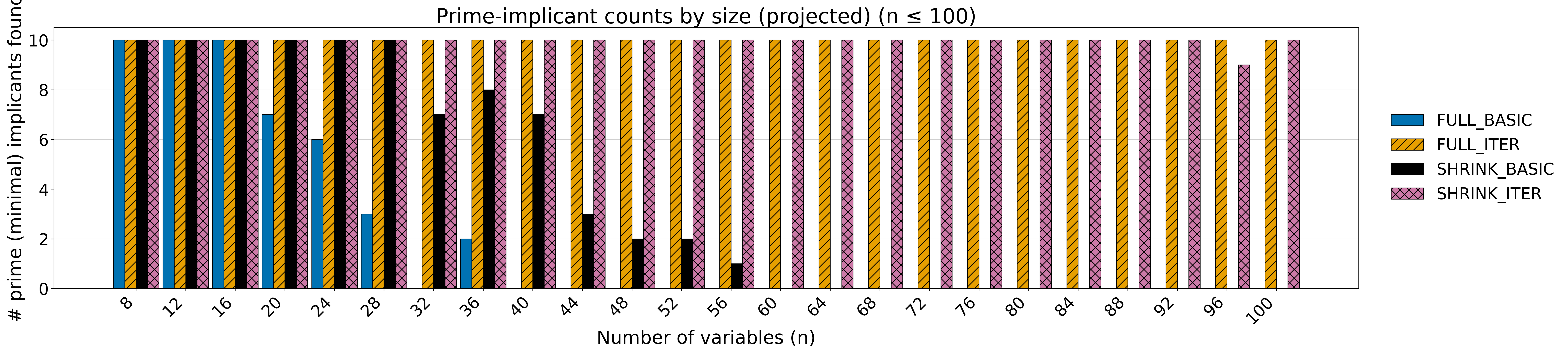}
    \caption{Prime (minimal) implicants found for $n\le 100$.
  \textbf{Top}: standard. \textbf{Bottom}: projected.}
    \label{fig:prime-projected}
  \end{minipage}
\end{figure*}

\begin{figure*}
\centering
    \includegraphics[width=0.65\linewidth]{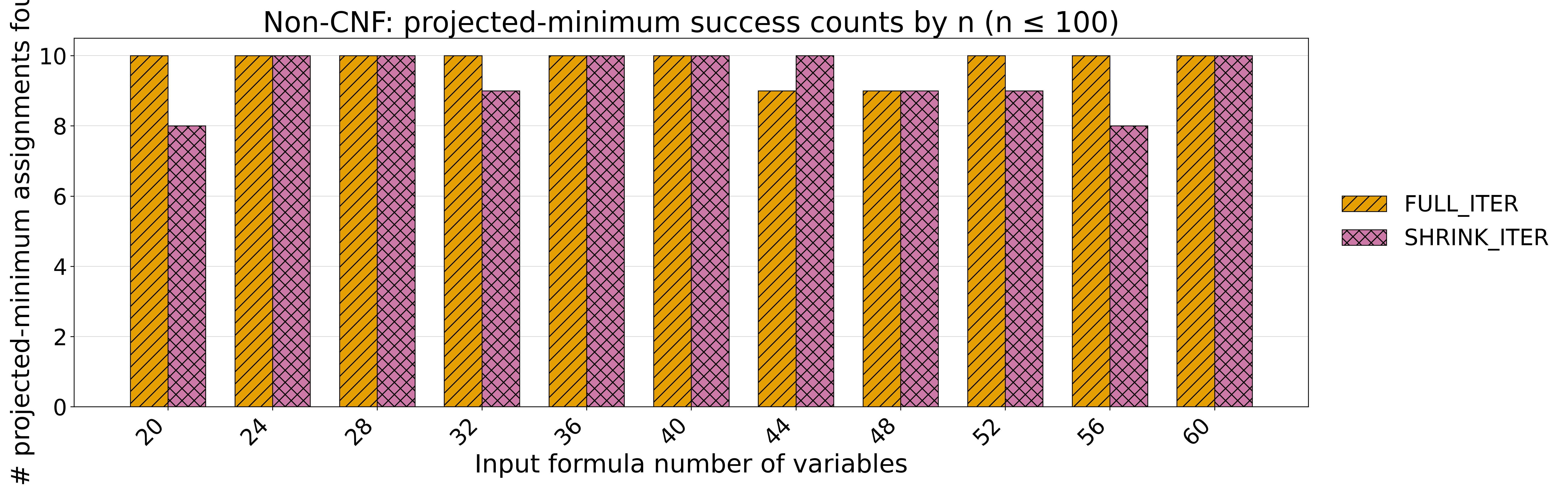}
    \caption{Projected minimum-cardinality assignments found for non-CNF formulae.}
    \label{fig:noncnf}
\end{figure*}



Because SA does not guarantee reaching a global optimum (as stated in Section~\ref{subsec:ising}), we explicitly post-check three properties of each returned solution: ($i$) implicant check, ($ii$) minimality check, and ($iii$) minimum-cardinality check. Figures~\ref{fig:min-projected} and~\ref{fig:prime-projected} report, for each $n \le 100$ (where exact minimum-cardinality verification is feasible), how many instances (out of $10$) satisfy these checks under each configuration.

First, all assignments retrieved under both \textsc{Basic} and \textsc{Iter}
satisfy the input formula and have energy below the $n\gamma$ threshold from
Corollary~\ref{cor:low-energy-feasible}, as intended. Baseline SA often recovers minimum implicants on small-to-medium instances, but the success rate decreases with $n$. \textsc{Iter} provides a modest but consistent improvement by further shrinking the best solution found in the initial run. The observed degradation should be interpreted as a search-quality effect, not an encoding-validity effect: there always exists a ground state decoded to a minimum-cardinality assignment, but baseline SA gets stuck to a local minimum instead.

In contrast, minimality is substantially easier to obtain, and iterative refinement markedly improves it. Once the polarity pattern is fixed, the refinement rounds act as a deletion phase, pushing additional variables into the don't-care state while preserving satisfiability, often eliminating redundant literals left by the initial run. On the 3-SAT instances, this improvement typically required only 1 or 2 refinement rounds.

We complement the random 3-SAT evaluation with the non-CNF benchmarks using \textsc{Iter}. Figure~\ref{fig:noncnf} reports projected QUBO minimization on the original variables $P$ after CNF conversion of the input problems. More refinement rounds are typically required to reach convergence (usually 3 to 5), plausibly due to the more complex induced energy landscape. Even with many auxiliary variables (about 65--70\% of all variables per instance), the method frequently recovers minimum-cardinality assignments on $P$. This behavior is consistent with Section~\ref{sec:projected}: for the NNF + Plaisted--Greenbaum pipeline, short assignments over the original variables admit a uniform single completion over the auxiliary variables.
Thus, even with a baseline QUBO optimizer, sampled solutions often decode directly to minimum-cardinality implicants, showing that the target optimum is already practically reachable.

\section{Conclusion and Future Work}
\label{sec:conclusion}

We introduce a new quadratic Ising/QUBO encoding of SAT aimed at computing short satisfying partial assignments. Each propositional variable is represented by a pair of polarity spins, inducing a three-valued semantics, where setting both polarities to $0$ corresponds to a don't-care state, and a sparsity term biases solutions toward small supports. We identified a penalty regime in which ground states correspond to satisfying assignments with minimum-cardinality support, and we showed that the same encoding can be used to shrink a given SAT model by fixing its polarity pattern, achieving minimum-cardinality implicants.

Two future work directions are most immediate. First, investigating alternative Ising approaches, such as quantum annealers and specialized Ising machines, may yield shorter assignments with higher probabilities. This also requires studies of hardware-aware variants that respect real-device topologies \cite{zbinden2020embedding}. Second, the encoding can be used as a minimization subroutine inside hybrid AllSAT workflows, replacing implicant-shrinking steps on harder instances or in recent network reliability encodings \cite{PAREDES2019106472}.

\bibliography{lipics-v2021-sample-article}

\newpage
\appendix

\section{Appendix}

\subsection{Minimum Implicant Projection Proof}
\label{appendix:theorem}

\begin{theorem}[Minimum-cardinality projected satisfying partial assignments]
\label{thm:projected-minimal-partial-models}
Let $F$ be a satisfiable CNF formula over variables $\mathcal V=P\cup H$, where
$P$ is the set of visible variables and $H$ the set of hidden variables. Assume there exists at least one projected implicant over $P$, and assume the weights satisfy
\[
\lambda > |P|\,\gamma
\qquad\text{and}\qquad
M > |P|\,\gamma.
\]
Assume moreover that the formula $F$ is structured such that for every partial assignment $\pi$ over $P$,
\[
\pi \models \exists H.\,F
\quad\Longleftrightarrow\quad
\exists \sigma \text{ over } H \text{ such that } \pi \wedge \sigma \models F.
\]
Then every ground state of the projected energy is ($i$) consistent, ($ii$) corresponds to a projected implicant of $F$ under $P$, and ($iii$) minimizes the number of active visible polarities among all projected implicants over $P$. Thus, under the above witness-completion assumption, the encoding reduces the computation of minimum-cardinality projected implicants to exact ground-state optimization.
\end{theorem}

\begin{proof}
Let $E^{\mathrm{proj}}_F$ be the projected energy function, where the sparsity term counts only active polarity spins on the visible variables $P$.
Let $\mathbf{z}^\star$ be a ground state (minimum-energy spin configuration) of $E^{\mathrm{proj}}_F$, and let $\pi^\star$ be the induced partial assignment over $P$ obtained from $\mathbf{z}^\star$ by reverting Equation~\ref{eq:semantic} on the visible variables.

By assumption, fix an arbitrary projected implicant $\tilde{\pi}$ over $P$, and by assumption let $\tilde{\sigma}$ be a hidden assignment over $H$ such that
\[
\tilde{\pi}\wedge \tilde{\sigma} \models F.
\]
Let $\tilde{\mathbf{z}}$ be the corresponding consistent polarity-spin encoding. For this configuration, all clause penalties can be set to $0$, hence
\[
E^{\mathrm{proj}}_F(\tilde{\mathbf{z}})= \gamma\cdot s_P(\tilde{\mathbf{z}}),
\]
where $s_P(\mathbf{z})$ denotes the number of active \emph{visible} polarity spins in $\mathbf{z}$.

\textbf{(i) Clause satisfaction / projected implication.}
Assume for contradiction that $\pi^\star$ is not a projected implicant of $F$.
Then
\[
\pi^\star \not\models \exists H.\,F.
\]
Hence, there exists no assignment $\sigma$ over $H$ such that
\[
\pi^\star \wedge \sigma \models F.
\]
Equivalently, no assignment to the hidden variables completes $\pi^\star$ to an ordinary implicant of $F$. By the semantics of the clause-penalty construction, this means that no full consistent spin configuration realizing $\pi^\star$ can make all clause penalties vanish. Therefore, any such configuration incurs at least one clause penalty, and hence its energy is at least $\lambda$. Thus
\[
E^{\mathrm{proj}}_F(\mathbf{z}^\star)\;\ge\;\lambda.
\]
On the other hand, for the witness-completed projected implicant $\tilde{\pi}\wedge \tilde{\sigma}$ we have
\[
E^{\mathrm{proj}}_F(\tilde{\mathbf{z}})=\gamma\cdot s_P(\tilde{\mathbf{z}})\le |P|\gamma.
\]
Since $\lambda>|P|\gamma$, then:
\[
E^{\mathrm{proj}}_F(\mathbf{z}^\star)\;>\;E^{\mathrm{proj}}_F(\tilde{\mathbf{z}}),
\]
contradicting that $\mathbf{z}^\star$ is a ground state. Hence, every ground state corresponds on $P$ to a projected implicant of $F$.

\textbf{(ii) Polarity consistency.}
Assume for contradiction that $\mathbf{z}^\star$ is inconsistent, i.e., there
exists a variable $A_i$ with both polarity spins active. Then the semantic consistency term contributes $M$, so
\[
E^{\mathrm{proj}}_F(\mathbf{z}^\star)\;\ge\;M.
\]
Using $M>|P|\gamma$ and again
\[
E^{\mathrm{proj}}_F(\tilde{\mathbf{z}})\le |P|\gamma,
\]
we derive
\[
E^{\mathrm{proj}}_F(\mathbf{z}^\star)\;>\;E^{\mathrm{proj}}_F(\tilde{\mathbf{z}}),
\]
contradicting that $\mathbf{z}^\star$ is a ground state. Therefore, every ground state is polarity-consistent.

\textbf{(iii) Minimum-cardinality among projected implicants.}
By (i) and (ii), $\mathbf{z}^\star$ is consistent, and its visible part $\pi^\star$ is a projected implicant of $F$. Hence, all clause penalties and semantic consistency penalties vanish, and
\[
E^{\mathrm{proj}}_F(\mathbf{z}^\star)=\gamma\cdot s_P(\mathbf{z}^\star),
\]
where $s_P(\mathbf{z})$ is the number of active visible polarity spins in $\mathbf{z}$.

Assume for contradiction that there exists another projected implicant $\pi'$ over $P$ with strictly fewer active visible polarities than $\pi^\star$.
By the witness-completion assumption, there exists an assignment $\sigma'$ over $H$ such that
\[
\pi' \wedge \sigma' \models F.
\]
Let $\mathbf{z}'$ be the corresponding consistent polarity-spin encoding. Then all clause penalties and consistency penalties vanish, so
\[
E^{\mathrm{proj}}_F(\mathbf{z}')=\gamma\cdot s_P(\mathbf{z}').
\]
Since $\pi'$ has strictly fewer active visible polarities than $\pi^\star$, we obtain
\[
s_P(\mathbf{z}')<s_P(\mathbf{z}^\star),
\]
and therefore
\[
E^{\mathrm{proj}}_F(\mathbf{z}')<E^{\mathrm{proj}}_F(\mathbf{z}^\star),
\]
contradicting that $\mathbf{z}^\star$ is a ground state.

Thus $\mathbf{z}^\star$ minimizes the number of active visible polarities among all projected implicants over $P$, i.e., it corresponds to a minimum-cardinality projected satisfying partial assignment.
\end{proof}

\subsection{Example of the Iterative Shrinking Procedure}
\label{app:iterative-shrinking-example}

We illustrate the iterative shrinking procedure on the formula
\[
F = (x_1 \lor \neg x_2 \lor x_3),
\]
starting from the total satisfying assignment
\[
\eta^{(0)} = x_1 \wedge \neg x_2 \wedge x_3 .
\]
In the polarity-splitting representation, $\eta^{(0)}$ is encoded as
\[
(p_1,n_1)=(1,0),\qquad
(p_2,n_2)=(0,1),\qquad
(p_3,n_3)=(1,0).
\]

To shrink $\eta^{(0)}$, we restrict the QUBO so that only literals consistent
with $\eta^{(0)}$ can be kept. Thus, the falsified polarities are fixed to zero:
\[
n_1=0,\qquad p_2=0,\qquad n_3=0,
\]
while the satisfied polarities $p_1, n_2, p_3$ remain free. Setting one of these spins to $1$ keeps the corresponding literal; setting it to $0$ drops the literal, leaving the variable unassigned.

Suppose that the first run of the Ising/QUBO optimizer returns
\[
p_1=1,\qquad n_2=1,\qquad p_3=0.
\]
This decodes to the partial assignment
\[
\eta^{(1)} = x_1 \wedge \neg x_2 .
\]
The literal $x_3$ has been removed, since $(p_3,n_3)=(0,0)$ now represents that
$x_3$ is unassigned. The assignment $\eta^{(1)}$ still satisfies $F$, because the
clause is already satisfied by $x_1$ and by $\neg x_2$, and has a lower energy than the initial total assignment.

Since the optimizer is heuristic, this first result is not necessarily a fixpoint. We therefore run the shrinking step again, now starting from
\[
\eta^{(1)} = x_1 \wedge \neg x_2 .
\]
The variable $x_3$ remains unassigned by fixing both $p_3$ and $n_3$ to 0, and the remaining compatible polarities are again allowed to be switched off:
\[
n_1=0,\qquad p_2=0,
\qquad
p_1,n_2\in\{0,1\}.
\]
A second run will be more prone to return a shorter assignment, e.g.: 
\[
p_1=1,\qquad n_2=0,
\]
which decodes to
\[
\eta^{(2)} = x_1 .
\]

This example shows how iterative shrinking works in the QUBO encoding. Polarity
spins are progressively switched from $1$ to $0$, thereby turning fixed literals
into unassigned variables. At each iteration, the optimizer may keep pruning
active polarities; however, pruning too aggressively would falsify some clauses
and incur a large increase in the objective value due to the clause penalty
term. Since Ising annealers' goal is to minimize the energy function, in proper weight regimes, we are guaranteed to move to shorter and shorter partial assignments.

Notice that this behavior is effective for encodings that explicitly represent the
unassigned state. In a standard SAT-to-Ising/QUBO encoding, each Boolean
variable is represented by a single spin, e.g., $z_i=1$ for True and $z_i=0$
for False. Hence, every optimizer output is necessarily a total assignment, and re-running the optimizer only produces the same total satisfying assignments, not a smaller implicant.

For example, the formula $F$ could be encoded by a penalty such as
\[
(1-z_1)z_2(1-z_3),
\]
which is zero whenever the clause is satisfied. If $z_1=1$, the values of $z_2$ and $z_3$ are apparently irrelevant to the satisfaction of the clause. However, since the annealer must assign a value to every spin, the resulting Boolean assignment is necessarily total: in particular, concrete values are still returned for $z_2$ and $z_3$. Thus, although the formula has logical don't-cares, the encoding does not represent them as unassigned variables. Adding a sparsity penalty on the single-spin encoding does not solve this issue: it can only bias variables toward one Boolean value, which differs from setting a variable to ``unassigned''. Consequently, an iterative shrinking procedure has no partial state to preserve across rounds: re-optimizing can only produce another total Boolean assignment, not progressively delete literals. This is why repeated optimization can shrink assignments in our polarity-splitting encoding, but not in standard two-valued encodings.

\end{document}